\documentclass[conference,twocolumn]{IEEEtran}
\IEEEoverridecommandlockouts
\usepackage{color}
\usepackage{amssymb}
\usepackage{amsmath}
\usepackage{graphicx,subfigure}
\usepackage{amsfonts}
\usepackage{cite}

\usepackage{color}
\usepackage{ dsfont }
\usepackage{hyperref}
\usepackage{algpseudocode}

\newtheorem{construction}{Construction}

\newtheorem{lemma}{Lemma}
\newtheorem{corollary}{Corollary}
\newtheorem{definition}{Definition}
\newtheorem{remark}{Remark}
\newtheorem{example}{Example}

\newcommand{\beqno}{ \begin{equation*} }
\newcommand{\eeqno}{ \end{equation*} }
\newcommand{\beq}{ \begin{equation} }
\newcommand{\eeq}{ \end{equation} }

\newcommand{\calM}{\mathcal{M}}
\newcommand{\calC}{\mathcal{C}}

\begin{document}
\title{Replication based storage systems with local repair}
\author{\IEEEauthorblockN{Oktay Olmez\IEEEauthorrefmark{1}\IEEEauthorrefmark{2} and Aditya Ramamoorthy\IEEEauthorrefmark{2}}
\IEEEauthorblockA{\IEEEauthorrefmark{1}Department of Mathematics,
Iowa State University,
Ames, Iowa 50011.\\}
\IEEEauthorblockA{\IEEEauthorrefmark{2}Department of Electrical and Computer Engineering,
Iowa State University,
Ames, Iowa 50011.\\
\{oolmez, adityar\}@iastate.edu}
\thanks{This work was supported in part by NSF grant CCF-1018148.}
}
\maketitle
\begin{abstract}
We consider the design of regenerating codes for distributed storage systems that enjoy the property of local, exact and uncoded repair, i.e., (a) upon failure, a node can be regenerated by simply downloading packets from the surviving nodes and (b) the number of surviving nodes contacted is strictly smaller than the number of nodes that need to be contacted for reconstructing the stored file.

Our codes consist of an outer MDS code and an inner fractional repetition code that specifies the placement of the encoded symbols on the storage nodes. For our class of codes, we identify the tradeoff between the local repair property and the minimum distance. We present codes based on graphs of high girth, affine resolvable designs and projective planes that meet the minimum distance bound for specific choices of file sizes.
\end{abstract}
\section{Introduction}
\label{sec:intro}

Large scale data storage systems that are employed in social networks, video streaming websites and cloud storage are becoming increasingly popular. In these systems the integrity of the stored data and the speed of the data access needs to be maintained even in the presence of unreliable storage nodes. This issue is typically handled by introducing redundancy in the storage system, through the usage of replication and/or erasure coding.
However, the large scale, distributed nature of the systems under consideration introduces another issue. Namely, if a given storage node fails, it need to be regenerated so that the new system continues to have the properties of the original system. It is of course desirable to perform this regeneration in a distributed manner and optimize performance metrics associated with the regeneration process.

In recent years, regenerating codes have been the subject of much investigation (see \cite{dimsurvey} and its references).
The principal idea of regenerating codes is to use subpacketization \cite{Dim}. In particular, one treats a given physical block as consisting of multiple packets (unlike the MDS code that stores exactly one packet in each node). Coding is now performed across the packets such that the file can be recovered by contacting a certain minimum number of nodes. In addition, one can regenerate a failed node by downloading appropriately coded data from the surviving nodes. 

A distributed storage system (henceforth abbreviated to DSS) consists of $n$ storage nodes, each of which stores $\alpha$ packets. In our discussion, we will treat these packets as elements from a finite field. Thus, we will equivalently say that each storage node contains $\alpha$ symbols (we use symbols and packets interchangeably throughout our discussion). A given user, also referred to as the data collector needs to have the ability to reconstruct the stored file by contacting any $k$ nodes; this is referred to as the maximum distance separability (MDS) property of the system. Suppose that a given node fails. The DSS needs to be repaired by introducing a new node. This node should be able to contact any $d \geq k$ surviving nodes and download $\beta$ packets from each of them for a total repair bandwidth of $\gamma = d \beta$ packets. Thus, the system has a repair degree of $d$, normalized repair bandwidth $\beta$ and total repair bandwidth $\gamma$. The new DSS should continue to have the MDS property. A large body of prior work (see for instance \cite{RSKR09,SuhR11,Dim} for a representative set) has considered constructions for functional and exact repair at both the minimum bandwidth regenerating (MBR) point where the repair bandwidth $\gamma$ is minimum and the minimum storage regenerating (MSR) point where the storage capacity $\alpha$ is minimum.


However, repair bandwidth is not the only metric for evaluating the repair process. It has been observed that the number of nodes that the new node needs to contact for the purposes of repair is also an important metric that needs to be considered. For either functional or exact repair (discussed above) the repair degree $d$ needs to be at least $k$. The notion of local repair introduced by \cite{gopalan12, papD12, oggierD11}, considers the design of DSS where the repair degree is strictly smaller than $k$. This is reasonable since contacting $k$ nodes allows the new node to reconstruct the entire file, assuming that the amount of data downloaded does not matter.

Much of the existing work in this broad area considers {\it coded} repair where the surviving nodes and the new node need to compute linear combinations for regeneration. It is well recognized that the read/write bandwidth of machines is comparable to the network bandwidth \cite{wiki}. 
Thus, this process induces additional undesirable delays \cite{jiekak_et_al12_preprint} in the repair process. The process can also be potentially memory intensive since the packets comprising the file are often very large (of the order of GB). 

In this work we consider the design of DSS that can be repaired in a local manner by simply downloading packets from the surviving nodes, i.e., DSS that have the exact and uncoded repair property.
\subsection{Background and Related Work}
The problem of local repair was first considered in references \cite{gopalan12,papD12,oggierD11}. Tradeoffs between locality and minimum distance, and corresponding code constructions were proposed in \cite{gopalan12} for the case of scalar codes ($\alpha =1$) and extended to the case of vector codes ($\alpha > 1$) in \cite{kamath12,rawat12,papD12}.
The design of DSS that have exact and uncoded repair and operate at the MBR point was first considered in the work of \cite{el10} and further constructions appeared in \cite{kooG11,olmezR12}. Codes for these systems are a concatenation of an outer MDS code and an inner fractional repetition code that specifies the placement of the encoded symbols on the storage nodes. In this work we consider the design of such codes that allow for local repair in the presence of one of more failures.

The work of \cite{kamath12, rawat12}, considers vector codes that allow local recovery in the presence of more than one failure. In their setting each storage node participates in a local code that has minimum distance greater than two. They present minimum distance bounds and corresponding code constructions that meet these bounds. The work of \cite{kamath12} on the design of MBR repair-by-transfer codes is most closely related to our work. However, as we shall see our constructions in Section \ref{sec:code_cons_bounds} are quite different from those that appear in \cite{kamath12} and allow for a larger range of code parameters. Moreover, as we focus on fractional repetition codes, our minimum distance bound is much tighter than the general case treated in \cite{kamath12, rawat12}.


\section{Problem Formulation}
\label{sec:problem_form}

The DSS is specified by parameters $(n,k,r)$ where $n$ - number of storage nodes, $k$ - number of nodes to be contacted for recovering the entire file and $r < k$ is the local repair degree, i.e., the number of nodes that an incoming node connects to for regenerating a failed node. The repair is performed by simply downloading packets from the existing nodes and is symmetric, i.e, the same number of packets are downloaded from each surviving node that is contacted.
It follows that we download $\beta_{loc} = \alpha/r$ packets from the surviving nodes. 

The proposed architecture for the system consists of an outer MDS code followed by an inner fractional repetition code. Specifically,
let the file that needs to be stored consist of $\calM$ symbols $x_1, \dots, x_{\calM}$. Suppose that these symbols are encoded using a $(\theta, \calM)$-MDS code to obtain encoded symbols $y_1, \dots, y_{\theta}$. The symbols $y_1, \dots, y_{\theta}$ are placed on the $n$ storage nodes, such that each symbol appears exactly $\rho$ times in the DSS. An example is illustrated in Fig. \ref{DSS-(15,4,4,2)}.



\begin{figure} [t]
\centering
\includegraphics[scale=0.25]{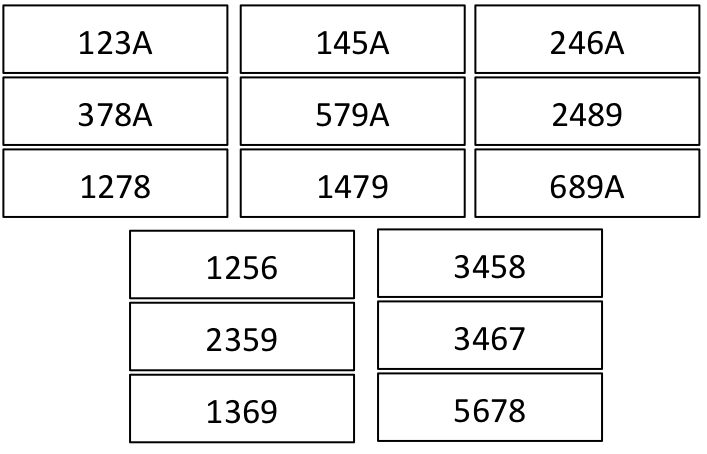}
\caption{The figure shows a DSS where $n=15,k=4,r=2,\theta=10,\alpha =4,\rho =6$. A node can be repaired locally by contacting the other two nodes in the same column. The system is resilient upto $5$ node failures.  On the other hand, local FR codes are resilient only a single node failure. Hence $\rho^{res} - 1=5$ and $\rho_{loc}^{res}=1$. Moreover, contacting any four nodes recovers at least $7$ distinct symbols so that the file size is $\calM = 7$. Therefore, the minimum distance of the code is $12$ for the filesize $\mathcal{M}=7$. }
\label{DSS-(15,4,4,2)}
\end{figure}

\begin{definition}
Let $\Omega = [\theta] = \{1, 2, \dots, \theta \}$ and $V_i, i = 1, \dots, d$ be subsets of $\Omega$. Let $V = \{V_1, \dots, V_d\}$ and consider $A \subset \Omega$ with $|A| = d\beta$. We say that $A$ is $\beta$-recoverable from $V$ if there exists $B_i \subseteq V_i$ for each $i= 1, \dots, d$ such that $B_i \subset A, |B_i| = \beta$ and $\displaystyle \cup_{i=1}^d B_i = A$.
\end{definition}
\begin{definition}
\label{defn:fr_code}
A fractional repetition (FR) code $\calC = (\Omega, V)$ for a $(n,k,d)$ DSS with repetition degree $\rho$ and normalized repair bandwidth $\beta = \alpha/d$  ($\alpha$ and $\beta$ are positive integers) is a set of $n$ subsets $V=\{V_1, \dots, V_n\}$ of a symbol set $\Omega = [\theta]$ with the following properties.
\begin{itemize}
\item[(a)] The cardinality of each $V_i$ is $\alpha$.
\item[(b)] Each element of $\Omega$ belongs to $\rho$ sets in $V$.
\item[(c)] Let $V^{surv}$ denote any $(n- \rho^{res} + 1)$ sized subset of $V$ and $V^{fail} = V \setminus V^{surv}$. Each $V_j \in V^{fail}$ is $\beta$-recoverable from some $d$-sized subset of $V^{surv}$.
\end{itemize}
\end{definition}
Note that we only consider FR codes without repeated storage nodes to avoid trivialities. It can be observed that $\rho^{res} \leq \rho$ is a measure of the resilience of the system to node failures, while still allowing exact and uncoded repair. We define the code rate of the system as $\frac{\calM}{n\alpha}$. 

For a FR code we define $a(\delta) = \min_{\{V_1, \dots, V_{\delta}\} \in \mathcal{I}} |\cup_{i=1}^\delta V_i|$, where $\mathcal{I}$ is the set of all $\delta$-sized subsets of $V$, i.e., $a(\delta)$ is the minimum number of symbols accumulated when union of $\delta$ storage nodes from $V$ is considered. We say that nodes $V_1, \dots, V_{\delta}$ cover at least $\zeta$ symbols if $|\cup_{i=1}^\delta V_i| \geq \zeta$. A FR code is in one-to-one correspondence with a 0-1 matrix of dimension $|V| \times |\Omega|$ (called the incidence matrix), where the $(i,j)$-th entry of the matrix is 1 if the $i$-th storage node contains the $j$-th symbol. Note that in an FR code we do not have any restriction on the repair degree $d$. 

\begin{definition} {\it Locally recoverable fractional repetition code.} Let $\mathcal{C} = (\Omega, V)$ be a FR code for a $(n,k,d)$ DSS, with repetition degree $\rho$ and normalized repair bandwidth $\beta= \alpha/d$. Let $r$ denote the local repair degree where $r < k$ and $r \leq d$. A node $V_i$ of $\mathcal{C}$ is said to be locally recoverable if there exists a set $W_i \subset V$ such that $V_i \in W_i$ and $V_i$ is $\alpha/r$-recoverable from $W_i \setminus V_i$. We call $W_i$ the local structure associated with node $V_i$. The FR code $\mathcal{C}$ is locally recoverable if all nodes in $V$ belong to at least one local structure.
\end{definition}

Let $\rho^{res}_{loc}$ denote the maximum number of node failures such that each failed node has at least one local structure in the set of surviving nodes. We call $\rho^{res}_{loc}$ the local failure resilience of the DSS.
We note that it is possible that the local structures themselves are FR codes; in this case we call them local FR codes. 


\begin{definition}{\it Minimum Distance.}
The minimum distance of a DSS denoted $d_{\min}$ is defined to be the size of the smallest subset of storage whose failure guarantees that the file is not recoverable from the surviving nodes.
\end{definition}
It is evident that $d_{\min} - 1 \geq \rho^{res} - 1 \geq \rho^{res}_{loc}$. In our constructions in Section \ref{sec:code_cons_bounds}, we will evaluate the different code designs on these parameters.




Bounds on the minimum distance of locally recoverable codes have been investigated in prior work. Specifically, \cite{gopalan12} considers the case of scalar ($\alpha = 1$) storage nodes and \cite{papD12} considers vector ($\alpha > 1$) storage nodes.
\begin{lemma}\label{local bound} Consider a locally recoverable DSS with parameters $n,k,r,\alpha$ with file size $\calM$ with minimum distance $d_{\min}$. Then,
\begin{align*}
d_{\min} \leq n- \left \lceil \frac{\calM}{\alpha}  \right \rceil -\left \lceil \frac{\calM}{r\alpha} \right \rceil + 2.
\end{align*}
\end{lemma}
Note that if a code is optimal with respect to Lemma $\ref{local bound}$, then the file can be recovered from $d_{\min}-1$ node erasures. This implies that $n-d_{\min} + 1 \geq k$. Therefore, a code is equivalently optimal if $  \left \lceil \frac{\calM}{\alpha}  \right \rceil + \left \lceil \frac{\calM}{r\alpha} \right \rceil \geq k +1$ and every set of $k$ nodes can reconstruct the file.

The minimum distance bound was tightened by \cite{kamath12,rawat12} when each storage node participates in a local code with minimum distance at least two; this allows for local recovery when there is more than one failure. However, for the class of codes that we consider, our bound (see Section \ref{sec:code_cons_bounds}) is tighter.


\section{Code constructions and bounds}
\label{sec:code_cons_bounds}


In this section, we present several code constructions and a minimum distance bound for a specific system architecture where the local structures are also FR codes.
\subsection{Codes for systems with $\rho_{loc}^{res} = 1$}
Our first construction is a class of codes which is optimal with respect to the bound provided in Lemma \ref{local bound} and allow local recovery in the presence of a single failure. Our construction leverages the properties of graphs with large girth.
\begin{definition} An undirected graph $\Gamma$ is called an $(s,g)$-graph if each vertex has degree $s$, and the length of the shortest cycle in $\Gamma$ is $g$.
\end{definition}
\begin{construction}
\label{grcons1}
Let $\Gamma=(V',E')$ be a $(s,g)$-graph with $|V'|=n$.
\begin{itemize}
\item[(i)] Arbitrarily index the edges of $\Gamma$ from 1 to $\frac{ns}{2}$.
\item[(ii)] Each vertex of $\Gamma$ corresponds to a storage node and stores the symbols incident on it.
\end{itemize}
\end{construction}
It can be observed that the above procedure yields an FR code $\mathcal{C} = (\Omega, V)$ with $n$ storage nodes, parameters $\theta = \frac{ns}{2}$, $\alpha = s$ and $\rho = 2$. Upon single failure, the failed node can be regenerated by downloading one symbol each from the storage nodes corresponding to the vertices adjacent to it in $\Gamma$ (i.e., $\beta_{loc} = 1$); thus, $r = s$.

\begin{remark} We note that the work of \cite{el10} also used the above construction for MBR codes where the file size was guaranteed to be  at least $k\alpha -\binom{k}{2}$; however, they did not have the restriction that $\Gamma$ is a $(s,g)$-graph. As we discuss next, $(s,g)$-graphs allow us to construct locally recoverable codes and provide a better bound on the file size when $k \leq g$. We allow the system parameter $k$ to be greater than $d$, however in the work of \cite{el10}, they consider only the case $k\leq d$.
\end{remark}
\begin{lemma} \label{lemma:grcons1_coverage} Let $\mathcal{C} = (\Omega, V)$ be a FR code constructed by Construction \ref{grcons1}. If $s > 2$, and $k \leq g$, we have $|\cup_{i=1}^k V_i| \geq k(s-1)$ for any $V_i \in V, i = 1, \dots k$
\end{lemma}


\begin{proof}
Let $V_1, V_2,\cdots,V_{k-1}$ and $V_{k}$ be any $k$ nodes in our DSS, where $k \leq g$. We argue inductively. Note that $|V_1| = s > s-1$. Suppose that $|\cup_{i=1}^j V_i| \geq j(s-1) + \xi$ for $j < k$, where $\xi \leq j$ is the number of connected components formed by the nodes $V_1, \dots, V_j$ in $\Gamma$. Now consider $|\cup_{i=1}^{j+1} V_i|$ where $j+1 < k$. Note that since $j+1 < g$ there can be no cycle in $\cup_{i=1}^{j+1} V_i$. Thus, $V_{j+1}$ is connected at most once to each connected component in $\cup_{i=1}^j V_i$. Suppose that $V_{j+1}$ is connected to $\ell$ existing connected components in $\cup_{i=1}^j V_i$, where $0 \leq \ell \leq s$. Then, the number of connected components in $\cup_{i=1}^{j+1} V_i$ is $\xi - \ell + 1$ and the number of new symbols that it introduces is $s - \ell$. Therefore $|\cup_{i=1}^{j+1} V_i| = j(s-1) + \xi + s - \ell = (j+1)(s-1) + \xi -\ell + 1$. This proves the induction step.

Thus, $|\cup_{i=1}^{k-1} V_i| \geq (k-1)(s-1) + \xi_{k-1}$, where $\xi_{k-1}$ is the number of connected components formed by $V_1, \dots, V_{k-1}$. Now consider $\cup_{i=1}^{k} V_i$. Note that there can be a cycle introduced at this step if $k=g$. Now, if $\xi_{k-1} \geq 2$, it can be seen that $V_k$ can only connect to each of the $\xi_{k-1}$ connected components once, otherwise it would imply the existence of a cycle of length strictly less than $g$ in $\Gamma$. Thus, in this case $|\cup_{i=1}^{k} V_i| \geq k(s-1)$. On the other hand if $\xi_{k-1} = 1$, then $V_k$ can connect at most twice to this connected component. In this case again we can observe that $|\cup_{i=1}^{k} V_i| \geq k(s-1)$.
\end{proof}

\begin{lemma} Let $\Gamma=(V,E)$ be a $(s,g)$-graph with $|V|=n$ and $s>2$. If $g\geq k=as+b$ such that $s > b \geq a+1$, then $\mathcal{C}$ obtained  from $\Gamma$ by Construction \ref{grcons1} is optimal with respect to bound in Lemma \ref{local bound} when the file size $\calM=k(s-1)$.
\end{lemma}
\begin{proof} From Lemma \ref{lemma:grcons1_coverage}, any $k$ nodes cover at least $k(s-1)$ symbols. Thus, the code is optimal when the following holds.
$$\displaystyle \left \lceil\frac{k(s-1)}{s} \right\rceil+ \left\lceil\frac{k(s-1)}{s^2}\right\rceil  \geq k+1.$$ We have 
$$k(s-1)=(as+b)(s-1)=as^2+(b-a)s-b.$$
Since, $s > b\geq a+1$ the following holds.
$$\displaystyle \left\lceil\frac{k(s-1)}{s}\right\rceil=\left\lceil\frac{as^2+(b-a)s-b}{s}\right\rceil=as+(b-a),$$
and
$$\displaystyle  \left\lceil\frac{as^2+(b-a)s-b}{s^2}\right\rceil=\left\lceil a+ \frac{(b-a)s-b}{s^2}\right\rceil\geq a+1 .$$

\end{proof}
\begin{corollary}\label{corollary_constr_1}
Let $\Gamma=(V,E)$ be a $(s,g)$-graph with $|V|=n$ and $s>2$. If $g\geq s+2$, then $\mathcal{C}$  obtained  from $\Gamma$ by Construction \ref{grcons1} is optimal with respect to the bound in Lemma \ref{local bound} for file size $\calM=s^2+s-2$.
\end{corollary}
It can be observed that in the specific case of $s=2$, applying Construction \ref{grcons1} results in a DSS where the union of any $k$ nodes has at least $k+1$ symbols. We now discuss some examples of codes that can be obtained  from our constructions.

\begin{remark}
Sachs \cite{Sachs} provided a construction which shows that for all $s,g \geq 3$, there exists a $s$-regular graph of girth $g$. Also, explicit constructions of graphs with arbitrarily large girth are known \cite{Lazebnik95}. Using these we can construct infinite families of optimal locally recoverable codes. 
\end{remark}

\begin{example}
The Petersen graph on 10 vertices and 15 edges can be shown to be $(3,5)$-graph. We label the edges $1, \dots, 10$ and $A, B, \dots, E$ in Fig. \ref{Ptrgraph}. Let the filesize $\calM = 3^2 +3 -2 = 10$; we use a $(15,10)$ outer MDS code. Applying Construction \ref{grcons1}, we obtain a DSS with parameters $n=10,k=5,\alpha=3,\rho=2,r=3$. From Corollary \ref{corollary_constr_1}, we observe that the DSS meets the minimum distance bound.


\begin{figure} [t]
\centering
\includegraphics[scale=0.25]{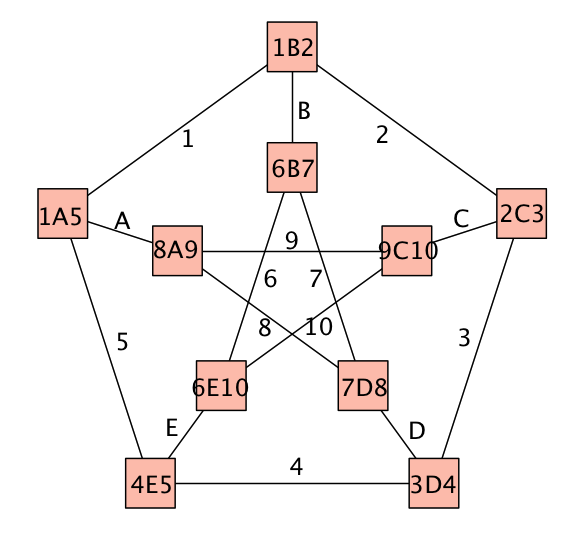}
\caption{{\small The figure shows the Petersen graph with its edges labeled from 1, \dots, 10 and $A, \dots, E$. Each vertex acts as a storage node and stores the symbols incident on it.}}
\label{Ptrgraph}
\vspace{-0.2in}
\end{figure}


\end{example}

An $(s,g)$-graph with the fewest possible number of vertices, among all $(s,g)$-graphs is called an $(s,g)$-cage and will result in the maximum code rate for our construction. For instance, the $(3,5)$-cage is the Petersen graph. 

We note here that bipartite cages of girth 6 were used to construct FR codes in \cite{kooG11} though these were not in the context of locally recoverable codes. It can be seen that Construction \ref{grcons1} can be extended in a straightforward way to larger filesizes.

\subsection{Codes for systems with $\rho_{loc}^{res} > 1$}
Our second class of codes are such that the local structures are also FR codes. The primary motivation for considering this class of codes is that they naturally allow for local recovery in the presence of more than one failure as long as the local FR code has a repetition degree greater than two. Thus, in these codes, each storage node participates in one or more local FR codes that allow local recovery in the presence of failures. For these classes of codes, we can derive the following tighter upper bound on the minimum distance (the proof appears in the Appendix) when the file size is larger than the number of symbols in one local structure.
\begin{lemma} \label{minimum distance} Let $\mathcal{C}$ be a locally recoverable FR code with parameters $(n,\theta,\alpha,\rho)$ where each node belongs to a local FR code with parameters $(n_{loc},\theta_{loc},\alpha,\rho_{loc})$. Suppose that the file size $\calM > \theta_{loc}$. Then, 
\[
\begin{split}
d_{min}& \leq \max \bigg{(}n- \left \lceil \frac{\calM\rho_{loc}}{\alpha} \right \rceil+\rho_{loc},\\
&n+n_{loc}+1- \left \lceil \frac{\calM\rho_{loc}+\theta_{loc}}{\alpha} \right \rceil \bigg{)}.
\end{split}
\]
\end{lemma}
The following corollary can be also be established (see Appendix). 

\begin{corollary}\label{mincor} Let $\mathcal{C}$ be a locally recoverable FR code with parameters $(n,\theta,\alpha,\rho)$ where each node belongs to a local FR code with parameters $(n_{loc},\theta_{loc},\alpha,\rho_{loc})$. Furthermore, suppose that $\mathcal{C}$ can be partitioned as the union of $\ell$ disjoint local FR codes. If the file size $\calM=t\theta_{loc}+\beta$ for some integer $1 \leq t < \ell$ and $\beta \leq \alpha$, we have $\displaystyle d_{min} \leq n-\left \lceil \frac{\calM\rho_{loc}}{\alpha} \right \rceil+\rho_{loc}.$
\end{corollary}

\begin{construction}\label{design}
Let $\mathcal{C} = (\Omega, V)$ be a FR code with parameters $(n,\theta,\alpha,\rho)$ such that any $\Delta$+1 nodes in $V$ cover $\theta$ symbols and for $V_i, V_j \in V$, we have $|V_i \cap V_j| \leq \beta$ when $i \neq j$. We construct a locally recoverable FR code $\bar{\mathcal{C}}$ by considering the disjoint union of $l (>1)$ copies of $\mathcal{C}$. Thus, $\bar{\mathcal{C}}$ has parameters $(ln,l\theta, \alpha, \beta)$.
\end{construction}
We call $\mathcal{C}$ the local FR code of $\bar{\mathcal{C}}$.
\begin{lemma} \label{lemma:cond_kron_opt} Let $ \bar{\mathcal{C}}$ be a code constructed by Construction \ref{design} for some $l>1$ such that the parameters of the local FR code satisfy $(\rho-1)\alpha\theta-(\theta+\alpha)(\Delta-1)\beta \geq 0$. Let the file size be $\calM= t\theta+\alpha$ for some $1 \leq t < l$. Then $\bar{\mathcal{C}}$ is optimal with respect to Corollary \ref{mincor}.
\end{lemma}
\begin{proof}
It is evident that $\bar{\mathcal{C}}$ is the disjoint union of $l$ local FR codes. Thus, the minimum distance bound here is $d_{min}\leq ln-\left \lceil \frac{(t\theta+\alpha)\rho}{\alpha}\right \rceil+\rho = (l-t)n.$ The code is optimal when any $tn+1$ nodes in $\bar{\mathcal{C}}$ cover at least  $\calM= t\theta+\alpha$ symbols. We show that this is the case below.

Let $a_i$ be the number of nodes that are chosen from the $i$-th local FR code and $X_i$ be the symbols covered by these $a_i$ nodes. Note that for any $1\leq i\leq l$ if $a_i \geq \Delta+1$, then $X_i = \theta$ (the maximum possible). Suppose there are $0 \leq t_1\leq t$ local FR codes that cover $\theta$ symbols. It can be seen that in this case it suffices to show that  $(t-t_1)n+1$ nodes cover at least  $(t-t_1)\theta+\alpha$ symbols. Here we can omit case of $t=t_1$, since our claim clearly holds in this situation. Suppose that these nodes belong to $s$ local FR codes, where $a_i \leq \Delta, i = 1, \dots, s$. By applying Corradi's lemma \cite{jukna11} we obtain
{\small
\[
\begin{split}
|X_i|\geq \frac{\alpha^2a_i}{\alpha+(a_i-1)\beta}
&\geq \frac{\alpha^2a_i}{\alpha+(\Delta-1)\beta}.\\
\end{split}
\]
}
This implies that
{\small
\[
\begin{split}
&\sum_{i=1}^{s}|X_i| \geq \sum_{i=1}^{s}\frac{\alpha^2a_i}{\alpha+(\Delta-1)\beta}\\
&=\frac{\alpha^2}{\alpha+(\Delta-1)\beta}\sum_{i=1}^{s}a_i\\
&=\frac{\alpha^2}{\alpha+(\Delta-1)\beta}((t-t_1)n+1)\\
&= \frac{(t-t_1)\theta\rho\alpha}{\alpha + (\Delta -1)\beta} + \frac{\alpha^2}{\alpha + (\Delta -1)\beta} \text{{\normalsize~(since $n\alpha = \theta\rho$)}}\\
&= (t-t_1)\theta  + \bigg{(}\frac{\rho\alpha}{\alpha + (\Delta -1)\beta} - 1\bigg{)}(t - t_1)\theta + \frac{\alpha^2}{\alpha + (\Delta -1)\beta}\\
&\geq (t-t_1)\theta + \frac{((\rho-1)\alpha - (\Delta -1)\beta)\theta + \alpha^2}{\alpha + (\Delta -1)\beta}\\
&\geq (t-t_1)\theta + \alpha \text{~{\normalsize (using the assumed conditions)}}
\end{split}
\]
}
\end{proof}
The above lemma can be used to generate several examples of locally recoverable codes with $\rho^{res}_{loc} > 1$. We discuss two examples below.
\begin{example}\label{affine resolvable}
In our previous work \cite{olmezR12} we used affine resolvable designs for the construction of FR codes that operate at the MBR point. Let $q$ be a prime power. These codes have parameters $\theta = q^m, \alpha = q^{m-1}, \rho = \frac{q^m-1}{q-1}$ and $n =q \rho$. Moreover the code is resolvable, i.e., we can vary the repetition degree by choosing an appropriate number of parallel classes. Suppose we choose the local FR code by including $q^{m-1}$ parallel classes. Thus, the parameters of the local FR code are $(n,\theta, \alpha,\rho) = (q^m, q^m, q^{m-1},q^{m-1})$. For this code it can be shown that $\Delta = q^m - q^{m-1}$ and that $\beta = q^{m-2}$. It can be observed that this local FR code satisfies the conditions of Lemma \ref{lemma:cond_kron_opt} when $m \geq 3$.

We construct a locally recoverable FR code $\bar{\mathcal{C}}$ by taking the disjoint union of $l>1$ of the above local FR codes. Thus, $\bar{\mathcal{C}}$ has parameters $(lq^m,lq^m,q^{m-1},q^{m-1})$. It can be seen that the code allows for local recovery in the presence of at most $q^{m-1} - 1$ failures, i.e., $\rho^{res}_{loc} = q^{m-1} - 1$. Let the file size be $\calM= tq^m+q^{m-1}$ for some $1\leq t<l$. Then
$ \bar{\mathcal{C}}$ is  optimal with respect to Corollary \ref{mincor}. 
\end{example}


\begin{example}\label{projective plane} A projective plane of order $q$ also forms a FR code $\mathcal{C} = (\Omega, V)$, where $\alpha = q+1$ and $\rho = q+1$. Furthermore, $|V_i \cap V_j| = 1$ if $i \neq j$ and each pair of symbols appears in exactly one node; this further implies that $\beta = 1$. A simple counting argument shows that $|\Omega| = \theta = q^2+q+1$ and $n = q^2 + q + 1$. It can be shown that $\mathcal{C}$ satisfies the conditions of Lemma \ref{lemma:cond_kron_opt} with $\Delta = q^2, \beta = 1$ since any $q^2+1$ nodes cover $q^2+q+1$ symbols.

We construct a locally recoverable FR code $\bar{\mathcal{C}}$ by taking $l>1$ copies of the code $\mathcal{C}$. So the code $\bar{\mathcal{C}}$ has parameters $(l(q^2+q+1),l(q^2+q+1),q+1,q+1)$. Let the file size be $\calM= t(q^2+q+1)+q+1$ for some $1\leq t<l$. Then, $\bar{\mathcal{C}}$ is  optimal with respect to Lemma  \ref{minimum distance} and has $\rho^{res}_{loc} = q$. An example is illustrated in Fig. \ref{Fano}.
\end{example}
It is worth noting that one can also obtain codes using the technique presented above by choosing the local FR code from several other structures including complete graphs and cycle graphs. Owing to space limitations, we cannot discuss all these examples here.
\begin{figure} [t]
\centering
\includegraphics[scale=0.3]{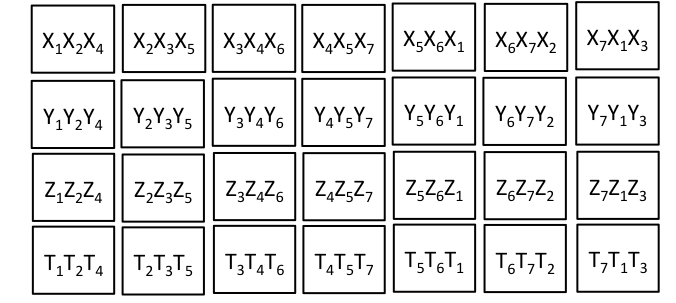}
\caption{The figure shows a DSS where $n=28,k=15,r=3,\theta=28,\alpha =3,\rho =3$ and each local FR code (the rows in the figure) is a projective plane of order $2$ which is also known as a Fano plane. Here $\rho^{res}- 1 = \rho_{loc}^{res}=2$. Any set of $15$ nodes cover at least $\mathcal{M}=17$ symbols. Thus, the minimum distance of the code is $14$ for the filesize $\mathcal{M}=17$.}
\label{Fano}\end{figure}

\section{Appendix}
\begin{figure}[t]
\begin{algorithmic}[1]
\State $S_0=\emptyset$, $i=1$
\While{$H(S_{i-1})<\calM$}
\State For each node $Y_j \in S_{i-1}$, identify an FR code $Pf_j$ (if it exists) such that $Y_j \in Pf_j, Pf_j \nsubseteq S_{i-1}$. If no such FR code exists, find an FR code that has no intersection with $S_{i-1}$ and set $Pf_1$ equal to it.
\begin{itemize} \item Let $b_j = |Pf_j \cap S_{i-1}|$. Let $j^* = \arg \max_j b_j$. \end{itemize}
\If {$\theta_{loc}-b_{j^*}+H(S_{i-1})<\calM$}
\State Set $S_i = S_{i-1}\cup Pf_{j^*}$.
\Else \If {there exists $A \subset  Pf_{j^*}$ such that $|S_{i-1} \cup A| > |S_{i-1}|$ and $H(S_{i-1} \cup A) < \calM$}
        \State Let $Pf^{'}_{j^*} = \arg \max_{A \subset  Pf_{j^*}} H(S_{i-1} \cup A) < \calM$. Set $S_i = S_{i-1}\cup Pf^{'}_{j^*}$.
        \Else
        \State Exit.
        \EndIf
\EndIf
\EndWhile\label{distance while}

\end{algorithmic}
\caption{Algorithm for finding the distance bound}\label{min_dist_algorithm}
\vspace{-0.1in}
\end{figure}

\begin{proof}
We will apply an algorithmic approach here (inspired by the one used in \cite{gopalan12}). Namely, we iteratively construct a set $\mathcal{S} \subset V$ so that $|\mathcal{S}| < \calM$. The minimum distance bound is then given by $n - |\mathcal{S}|$. Our algorithm is presented in Fig. \ref{min_dist_algorithm}. Towards this end, let $S_{i}$ and $H(S_i)$ represent the number of nodes and the number of symbols included at the end of the $i$-th iteration. Furthermore, let $s_i=|S_i|-|S_{i-1}|$ and $h_i=|H(S_i)|-|H(S_{i-1})|$, represent the corresponding increments between the $i-1$-th and the $i$-the iteration.
We divide the analysis into two cases.
\begin{itemize}
\item Case 1: [The algorithm exits without ever entering line 8.] Note that we have $1 \leq s_i \leq n_{loc}$ and $h_i\leq \theta_{loc} -a(n_{loc}-s_i)$ where $a(n_{loc}-s_i)$ is the minimum number of symbols covered by $(n_{loc} - s_i)$ nodes in the local FR code and hence the minimum size of $|Pf_{j^*} \cap S_{i-1}|$. By considering the bipartite graph representing the local FR code it can be seen that $\displaystyle a(n_{loc}-s_i) \geq \frac{(n_{loc}-s_i)\alpha}{\rho_{loc}}$ .
    Thus, we have
$$\displaystyle \theta_{loc}-a(n_{loc}-s_i) \leq \theta_{loc}-\frac{n_{loc}\alpha-s_i\alpha}{\rho_{loc}}=\frac{s_i\alpha}{\rho_{loc}}.$$ Suppose that the algorithm runs for $l$ iterations and exits on the $l+1$ iteration. Then
$$\displaystyle \sum_{i=1}^ls_i \geq \frac{\rho_{loc}}{\alpha}\sum_{i=1}^l h_i.$$ Since the algorithm exits without ever entering line $8$, it is unable to accumulate even one additional node. Hence
\begin{align*} \sum_{i=1}^l h_i &\geq \calM-\alpha, \text{~which implies}\\
\sum_{i=1}^l s_i &\geq \left \lceil \frac{\rho_{loc}}{\alpha}(\calM-\alpha) \right \rceil  \text{~by the integer constraint}.
\end{align*}
Thus, the bound on the minimum distance becomes
$$\displaystyle d_{min}\leq n- \left \lceil\frac{\rho_{loc}M}{\alpha} \right \rceil  +\rho_{loc}.$$
\item Case 2: [The algorithm exits after entering line 8.] Note that by assumption, $\calM > \theta_{loc}$. Suppose that the algorithm enters line $5$, $l \geq 1$ times. Now we have $\displaystyle \sum_{i=1}^l h_i \geq \calM-\theta_{loc}$, otherwise we could include another local structure. Hence we need to add nodes so that strictly less than $\displaystyle \calM - \sum_{i=1}^l h_i $ symbols are covered. It can be seen that we can include at least $\displaystyle \left \lceil\frac{\calM-\sum_{i=1}^l h_i}{\alpha} \right \rceil -1$  more nodes. Therefore, the total number of nodes
\begin{align*}
&\geq \frac{\rho_{loc}}{\alpha}\sum_{i=1}^l h_i+ \left \lceil\frac{\calM-\sum_{i=1}^l h_i}{\alpha} \right \rceil  - 1\\
&\geq \frac{\rho_{loc}-1}{\alpha}(\calM-\theta_{loc})+\frac{\calM}{\alpha}-1\\
&=\frac{\calM\rho_{loc}+\theta_{loc}}{\alpha}-n_{loc}-1.
\end{align*}
Therefore, we have the following minimum distance bound
$$\displaystyle d_{min} \leq n+n_{loc}+1- \left \lceil \frac{\calM\rho_{loc}+\theta_{loc}}{\alpha} \right \rceil. $$
The final bound is obtained by taking the maximum of the two bounds obtained above.
\end{itemize}

The proof of Corollary \ref{mincor} follows by observing that when the code consists of $l$ disjoint local FR codes and the file size $\calM = t\theta_{loc} + \beta$, where $0 \leq \beta \leq \alpha$, the algorithm in Fig. \ref{min_dist_algorithm} never enters line 8.
\end{proof}


\begin{thebibliography}{10}
\providecommand{\url}[1]{#1}
\csname url@samestyle\endcsname
\providecommand{\newblock}{\relax}
\providecommand{\bibinfo}[2]{#2}
\providecommand{\BIBentrySTDinterwordspacing}{\spaceskip=0pt\relax}
\providecommand{\BIBentryALTinterwordstretchfactor}{4}
\providecommand{\BIBentryALTinterwordspacing}{\spaceskip=\fontdimen2\font plus
\BIBentryALTinterwordstretchfactor\fontdimen3\font minus
  \fontdimen4\font\relax}
\providecommand{\BIBforeignlanguage}[2]{{%
\expandafter\ifx\csname l@#1\endcsname\relax
\typeout{** WARNING: IEEEtran.bst: No hyphenation pattern has been}%
\typeout{** loaded for the language `#1'. Using the pattern for}%
\typeout{** the default language instead.}%
\else
\language=\csname l@#1\endcsname
\fi
#2}}
\providecommand{\BIBdecl}{\relax}
\BIBdecl

\bibitem{dimsurvey}
A.~Dimakis, K.~Ramchandran, Y.~Wu, and C.~Suh, ``A survey on network codes for
  distributed storage,'' \emph{Proceedings of the IEEE}, vol.~99, no.~3, pp.
  476 --489, 2011.

\bibitem{Dim}
A.~Dimakis, P.~Godfrey, Y.~Wu, M.~Wainwright, and K.~Ramchandran, ``Network
  coding for distributed storage systems,'' \emph{IEEE Trans. on Info. Th.},
  vol.~56, no.~9, pp. 4539 --4551, Sept. 2010.

\bibitem{RSKR09}
K.~Rashmi, N.~Shah, P.~Kumar, and K.~Ramchandran, ``Explicit construction of
  optimal exact regenerating codes for distributed storage,'' in \emph{47th
  Annual Allerton Conference on Communication, Control, and Computing}, 2009,
  pp. 1243 --1249.

\bibitem{SuhR11}
C.~Suh and K.~Ramchandran, ``Exact-repair mds code construction using
  interference alignment,'' \emph{IEEE Trans. on Info. Th.}, vol.~57, no.~3,
  pp. 1425 --1442, 2011.

\bibitem{gopalan12}
P.~Gopalan, C.~Huang, H.~Simitci, and S.~Yekhanin, ``On the locality of
  codeword symbols,'' \emph{IEEE Trans. on Info. Th.}, vol.~58, no.~11, pp.
  6925 --6934, 2012.

\bibitem{papD12}
D.~Papailiopoulos and A.~Dimakis, ``Locally repairable codes,'' in \emph{IEEE
  Intl. Symposium on Info. Th.}, 2012, pp. 2771 --2775.

\bibitem{oggierD11}
F.~Oggier and A.~Datta, ``Self-repairing homomorphic codes for distributed
  storage systems,'' in \emph{INFOCOM, 2011 Proceedings IEEE}, april 2011, pp.
  1215 --1223.

\bibitem{wiki}
``Wikipedia: List of device bit rates, available at
  http://en.wikipedia.org/wiki/list\textunderscore of \textunderscore device
  \textunderscore bandwidths.''

\bibitem{jiekak_et_al12_preprint}
S.~Jiekak, A.-M. Kermarrec, N.~L. Scouarnec, G.~Straub, and A.~V. Kempen,
  ``{Regenerating Codes: A System Perspective},'' 2012 [Online] Available:
  http://arxiv.org/abs/1204.5028.

\bibitem{kamath12}
G.~M. Kamath, N.~Prakash, V.~Lalitha, and P.~V. Kumar, ``{Codes with Local
  Regeneration},'' 2012 [Online] Available: http://arxiv.org/abs/1211.1932.

\bibitem{rawat12}
A.~K. Rawat, O.~O. Koyluoglu, N.~Silberstein, and S.~Vishwanath, ``{Optimal
  Locally Repairable and Secure Codes for Distributed Storage Systems},'' 2012
  [Online] Available: http://arxiv.org/abs/1210.6954.

\bibitem{el10}
S.~E. Rouayheb and K.~Ramchandran, ``Fractional repetition codes for repair in
  distributed storage systems,'' in \emph{48th Annual Allerton Conference on
  Communication, Control, and Computing}, 2010, pp. 1510 --1517.

\bibitem{kooG11}
J.~Koo and J.~Gill, ``Scalable constructions of fractional repetition codes in
  distributed storage systems,'' in \emph{49th Annual Allerton Conference on
  Communication, Control, and Computing}, 2011, pp. 1366 --1373.

\bibitem{olmezR12}
O.~Olmez and A.~Ramamoorthy, ``Repairable replication-based storage systems
  using resolvable designs,'' in \emph{50th Annual Allerton Conference on
  Communication, Control, and Computing}, 2012.

\bibitem{Sachs}
H.~Sachs, ``Regular graphs with given girth and restricted circuits,''
  \emph{Journal of The London Mathematical Society-second Series}, vol. s1-38,
  no.~1, pp. 423--429, 1963.

\bibitem{Lazebnik95}
F.~Lazebnik and V.~A. Ustimenko, ``Explicit construction of graphs with an
  arbitrary large girth and of large size,'' \emph{Discrete Applied
  Mathematics}, vol.~60, no.~1, pp. 275--284, 1995.

\bibitem{jukna11}
S.~Jukna, \emph{{Extremal combinatorics: with applications in computer
  science}}.\hskip 1em plus 0.5em minus 0.4em\relax Springer, 2011.

\end{thebibliography}
\end{document}